\documentclass[11pt,a4note]{article}
\pdfoutput=1

\usepackage{hyperref}
\usepackage{fullpage}
\usepackage{amssymb}
\usepackage{mathtools}
\usepackage{amsthm}
\usepackage{eucal}
\usepackage{dsfont}
\usepackage[T1]{fontenc}
\usepackage{lmodern}
\usepackage[stretch=10,shrink=10]{microtype}
\usepackage{color,soul}

\allowdisplaybreaks[1]

\def\F{\mathbb{F}}

\def\O{\CMcal{O}}

\theoremstyle{plain}
\newtheorem{theorem}{Theorem}[section]
\newtheorem{lemma}[theorem]{Lemma}

\theoremstyle{definition}
\newtheorem{definition}[theorem]{Definition}

\newtheorem{fact}[theorem]{Fact}
\newtheorem{conjecture}[theorem]{Conjecture}

\newcommand{\reals}{{\mathbb R}}
\newcommand {\set} [1] {\ensuremath{ \left\lbrace #1 \right\rbrace }}

\newcommand{\codim}[1]{\mathrm{codim}\br{#1}}

\newcommand {\br} [1] {\ensuremath{ \left( #1 \right) }}

\newcommand {\norm} [1] {\ensuremath{ \left\| #1 \right\| }}

\newcommand {\abs} [1] {\ensuremath{ \left| #1 \right| }}

\newcommand {\defeq} {\ensuremath{ \stackrel{\mathrm{def}}{=} }}

\newcommand{\rk}[1]{\mathrm{rk}\br{#1}}

\newcommand{\pdt}[1]{\mathrm{D}_{\oplus}\br{#1}}
\newcommand{\cc}[2]{\mathrm{CC}^{\br{#1}}\br{#2}}
\newcommand{\pcert}[1]{\mathrm{C}_{\oplus}\br{#1}}
\newcommand{\minpcert}[2]{\mathrm{C}_{\oplus,\mathrm{min}}^{#1}\br{#2}}

\newcommand {\email} [1] {\href{mailto:#1}{\texttt{#1}}}

\newcommand {\mytitle} {Parity Decision Tree Complexity and $4$-Party Communication Complexity of XOR-functions Are Polynomially Equivalent}

\newcommand{\Penghui}{Penghui Yao}
\newcommand{\CWI}{CWI, Amsterdam}

\newcommand {\authorblock} [3] {
	\begin{minipage}[t]{0.3\linewidth}
		\centering
		{#1}\\[0.8ex]
		 {#2}\\[-0.7ex]
		\email{#3}
	\end{minipage}\vspace{1ex}
}

\hypersetup{
	pdfstartview={FitH},
	pdfdisplaydoctitle={true},
	breaklinks={true},
	bookmarksopen={true},
	bookmarksnumbered={false},
	pdftitle={\mytitle},
}

\begin{document}

\begin{titlepage}
\title{\textbf{\mytitle}\\[2ex]}

\author{
	\authorblock{\Penghui\thanks{Supported by the European Commission FET-Proactive project Quantum Algorithms (QALGO) 600700.}}{\CWI}{phyao1985@gmail.com} 
	}

\clearpage\maketitle
\thispagestyle{empty}

\begin{abstract}
In this note, we study the relation between the {\em parity decision tree complexity} of a boolean function $f$, denoted by $\pdt{f}$, and the {\em $k$-party number-in-hand multiparty communication complexity} of the XOR functions $F(x_1,\ldots, x_k)\defeq f(x_1\oplus\cdots\oplus x_k)$, denoted by $\cc{k}{F}$. It is known that $\cc{k}{F}\leq k\cdot\pdt{f}$ because  the players can simulate the parity decision tree that computes $f$.  In this note, we show that  \[\pdt{f}\leq\O\br{\cc{4}{F}^5}.\] Our main tool is a recent result from additive combinatorics due to Sanders~\cite{Sanders:2012}. As $\cc{k}{F}$ is non-decreasing as $k$ grows, the parity decision tree complexity of $f$ and the communication complexity of the corresponding $k$-argument XOR functions are polynomially equivalent whenever $k\geq 4$.

Remark: After the first version of this paper was finished, we discovered that Hatami and Lovett had already discovered the same result a few years ago, without writing it up.
\end{abstract}
\end{titlepage}

\section{Introduction}

\textbf{Communication complexity and the Log-Rank conjecture for XOR functions} Communication complexity quantifies the minimum amount of communication needed for computation when inputs are distributed among different parties~\cite{Yao:1979:CQR:800135.804414, Kushilevitz96}.  In the model of two-party communication, Alice and Bob hold inputs $x$ and $y$, respectively, and they are supposed to compute the value of a function $F(x,y)$ using as little communication as possible. One of the central problems in communication complexity is the {\em Log-Rank conjecture}.  The conjecture proposed by Lov\'asz and Saks in~\cite{Lovasz:1988} asserts that the communication complexity of $F$ and $\log\mathrm{rank}\br{M_F}$ are polynomially equivalent for any $2$-argument total boolean function $F$, where $M_F=[F(x,y)]_{x,y}$ is the communication matrix of $F$. Readers may refer to~\cite{TsangWXZ:2013} for more discussion on the conjecture. The conjecture is notoriously hard to attack.  It was shown over 30 years~\cite{Mehlhorn:1982} that $\log\mathrm{rank}\br{M_F}$ is a lower bound on the deterministic communication complexity of $F$. The state of the art is $$\cc{2}{F}\leq\O\br{\sqrt{\mathrm{rank}\br{M_F}}\log\mathrm{rank}\br{M_F}},$$ 
where $\cc{2}{F}$ stands for the two-party deterministic communication complexity of $F$. It is from a recent breakthrough due to Lovett~\cite{Lovett:2014a}. The largest gap between $\cc{2}{F}$ and $\log\mathrm{rank}\br{M_F}$ is $\cc{2}{F}\geq\Omega\br{\log\mathrm{rank}\br{M_F}^{\log_36}}$ due to Kushilevitz in~\cite{Nisan:1995} .

In~\cite{Zhang:2010}, Zhang and Shi initiated the study the Log-Rank conjecture for a special class of functions called {\em XOR functions}.

\begin{definition}\label{def:xorfunction}
We say a $k$-argument function $F:\br{\set{0,1}^n}^k\rightarrow\set{0,1}$ is an XOR-function if there exists a function $f:\set{0,1}^n\rightarrow\set{0,1}$ such that $F(x_1,\ldots, x_k)=f(x_1\oplus \ldots\oplus  x_k)$ for any $x_1,\ldots, x_k\in\set{0,1}^n$, where $\oplus$  is bitwise xor.
\end{definition}

XOR functions include many important examples, such as Equality and Hamming distance. The communication complexity of $XOR$ functions has been studied extensively in the last decade ~\cite{Zhang:2009,Lee:2010,Montanaro:2010,Leung:2011,TsangWXZ:2013,Zhang:2014}. 
A nice feature of XOR functions is that the rank of the communication matrix $M_F$ is exactly the {\em Fourier sparsity} of $f$. 
\begin{fact}\cite{BernasconiC:1999}\label{fact:rankisspasity}
For XOR function $F(x,y)\defeq f(x\oplus y)$, it holds that $\mathrm{rank}\br{M_F}=\norm{\hat{f}}_0$, where $\norm{\hat{f}}_0$ is the Fourier sparsity of $f$ (see Section~\ref{Sec:preliminaries} for the definition) and $M_F$ is the communication matrix of $F$. 
\end{fact}

Therefore, the Log-Rank conjecture for XOR functions is equivalent to the question  whether there exists a protocol computing $F$ with communication $\log^{\O(1)}\norm{\hat{f}}_0$. However, the Log-Rank conjecture is still difficult for this special class of functions.  One nice approach  proposed in~\cite{Zhang:2009} is to design a parity decision tree (PDT) to compute $f$. PDTs allow query the parity of any subset of input variables. For any $k$-argument XOR function $F$ given in Definition~\ref{def:xorfunction}, we can construct a communication protocol by simulating the PDT for $f$, with communication $k$ times the PDT complexity of $f$. It is therefore sufficient to show that $\pdt{f}\leq\log^{\O(1)}\norm{\hat{f}}_0$.  Using such an approach, the Log-Rank Conjecture has been established for several subclasses of XOR functions~\cite{Zhang:2009,Montanaro:2010,TsangWXZ:2013}. 

One question regarding this approach is whether $\pdt{f}$ and $\cc{2}{F}$ are polynomially equivalent.  Is it possible to design a protocol for $F$ much more efficient than simulating the parity decision tree of $f$? 

\begin{conjecture}\label{conj:pdtvscc2}
There is a constant $c$ such that $\cc{2}{F}=\O\br{\pdt{F}^c}$ for any boolean function $f:\set{0,1}^n\rightarrow\set{0,1}$ and $F(x,y)\defeq f(x\oplus y)$.
\end{conjecture}

If this holds, then the Log-Rank conjecture for XOR-functions is equivalent to a question in parity decision tree. Namely, $\pdt{f}\leq\mathrm{poly}\log\br{\norm{\hat{f}}_0}$. In this note, we prove a weaker variant of the above conjecture. Given a total boolean function $f$, we may also consider the communication complexity of the $k$-argument XOR-function $F_k(x_1,\ldots, x_k)\defeq f(x_1\oplus\cdots\oplus x_k)$ in the model of {\em number-in-hand multiparty communication}, which is denoted by $\cc{k}{F_k}$. It is easy to see that $\cc{2}{F_2}\leq\cc{3}{F_3}\leq\ldots$ and $\cc{k}{F_k}\leq k\cdot\pdt{f}$.  Our main result in this note is that $\cc{k}{F_k}$ and $\pdt{f}$ are polynomially equivalent whenever $k\geq 4$.

\begin{theorem}\label{thm:maintheorem}
For any boolean function $f:\set{0,1}^n\rightarrow\set{0,1}$, we define a $4$-argument XOR function by $F(x_1,x_2,x_3,x_4)=f\br{x_1\oplus x_2\oplus x_3\oplus x_4}$. It holds that 
\[\pdt{f}\leq\O\br{\cc{4}{F}^5}.\]
\end{theorem}

\subsection*{Our techniques}

To show the main theorem, it suffices to construct an efficient PDT for $f$ if the communication complexity of $F$ is small. We adapt a protocol introduced by Tsang et al.~\cite{TsangWXZ:2013}. The main step is to exhibit a large monochromatic affine subspace for $f$ if the communication complexity of $F$ is small. To this end, we adapt the quasipolynomial Bogolyubov-Ruzsa lemma~\cite{Sanders:2012}, which says that $4A\defeq{A+A+A+A}$ contains a large subspace if $A\subseteq\F_2^n$ is large. 

\subsection*{Related work}

A large body of work has been devoted to the Log-Rank conjecture for XOR functions since it was proposed in~\cite{Zhang:2009}. After almost a decade of efforts, the conjecture has been established for several classes of XOR function, such as symmetric functions~\cite{Zhang:2009}, monotone functions and linear threshold functions~\cite{Montanaro:2010}, constant $\F_2$-degree functions~\cite{TsangWXZ:2013}.  

A different line of work close to ours is the {\em simulation theorem} in~\cite{RazM:1999,Zhang:2009, Sherstov:2010, LovettMWZ:2014, Pitassi:2015}. They study the relation between the (regular) decision tree complexity of function $f$ and the communication complexity of $f\circ g^n$ where $g$ is a $2$-argument function of small size.   The simulation theorem asserts that the optimal protocol for $f\circ g^n$ is to simulate the decision tree that computes $f$ if $g$ is a hard function.  Simulation theorems have been established  in various cases, when $g$ is bitwise AND or OR~\cite{Sherstov:2010}, Inner-Product~\cite{LovettMWZ:2014}, Index Function~\cite{RazM:1999,Pitassi:2015}. Our work gives a new simulation theorem when $g$ is an $XOR$ function.

After this work was put online, the author was informed that Hatami and Lovett discovered Theorem~\ref{thm:maintheorem} (using the same idea) a couple of years ago without writing it up. Since our work is independent of theirs, we believe it is worth giving a complete proof to the main theorem.

\section{Preliminaries}\label{Sec:preliminaries}

All logarithms in this note are base $2$. Given $x,y\in\set{0,1}^n$, we define the inner product $x\cdot y\defeq\sum_{i=1}^nx_iy_i\mod 2$.  For simplicity, we write $x+y$ for $x\oplus y$. 

\textbf{Complexity measures.}\quad
Given a  boolean function $f:\set{0,1}^n\rightarrow \set{0,1}^n$, it can be viewed as a polynomial in $\F_2$, and  $\deg_2(f)$ is used to represent its $\F_2$-degree.
\begin{definition}\label{def:paritycertificatecomplexity}
Given a  function $f:V\rightarrow\F_2$, where $V$ is an affine subspace of $\F_2^n$, the parity certificate complexity of $f$ on $x$ is defined to be 
\[\pcert{f,x}\defeq\min\set{\codim{H}~: ~ H\subseteq V~\text{is an affine subspace where $f$ is constant and}~x\in H}\]
where $\codim{H}\defeq\dim{V}-\dim{H}$.
The minimum parity certificate complexity for $b\in\set{0,1}$ is defined as
\[\minpcert{b}{f}\defeq\min_{x\in f^{-1}(b)}\pcert{f,x},\]
and $\minpcert{}{f}\defeq\min_x\pcert{f,x}$.
\end{definition}

\begin{definition}\label{def:rank}
Given a  boolean function $f:\set{0,1}^n\rightarrow \set{0,1}$. We view it as a polynomial in $\F_2$. The {\em linear rank} of $f$,  denoted $\rk{f}$, is the minimum integer $r$, such that $f$ can be expressed as $f=\sum_{i=1}^r l_if_i+ f_0$, where $\deg_2\br{l_i}=1$ for $1\leq i\leq r$ and $\deg_2(f_i)<\deg_2(f)$ for $0\leq i\leq r$.
\end{definition}

\begin{definition}\label{def:PDTcomplexity}
A parity decision tree (PDT) for a boolean function $f:\set{0,1}^n\rightarrow\set{0,1}$ is a tree with internal nodes associated with a subset $S\subseteq[n]$ and each leaf associated with an answer in $\set{0,1}$. To use a parity decision tree to compute $f$, we start from the root and follow a path down to a leaf. At each internal node, we query the parity of the bits with the indices in the associated set and follow the branch according to the answer to the query. Output the associated answer when we reach the leaf.  The deterministic parity decision tree complexity of $f$, denoted by $\pdt{f}$, is the minimum number of queries needed on a worst-case input by a PDT that computes $f$ correctly.
\end{definition}

\begin{definition}\label{def:communicationcomplexity}
In the model of number-in-hand multiparty communication, there are $k$ players $\set{P_1, \ldots, P_k}$ and a $k$-argument function $F:\br{\set{0,1}^n}^k\rightarrow\set{0,1}$. Player $P_i$ is given an $n$-bit input $x_i\in\set{0,1}^n$ for each $i\in [k]$. The communication is in the blackboard model. Namely, every message sent by a player is written on a blackboard visible to all players.  The communication complexity of $f$ in this model, denoted by $\cc{k}{F}$, is the least number of bits needed to be communicated to compute $f$ correctly.
\end{definition}
One way to design a protocol for the $k$-argument XOR-function $F(x_1, \ldots, x_k)\defeq f(x_1+\ldots+x_k)$  to simulate a parity decision tree that computes  $f$.

\begin{fact}
Let $f:\set{0,1}^n\rightarrow\set{0,1}$ be a boolean function and $F$ be the $k$-argument XOR function defined as $F(x_1, \ldots, x_k)\defeq f(x_1+\cdots+x_k)$. It holds that 
$\cc{k}{F}\leq k\cdot\pdt{f}$.
\end{fact}

\textbf{Fourier analysis.} For any real function $f:\set{0,1}^n\rightarrow \reals$, the Fourier coefficients are defined as $\hat{f}\br{s}\defeq\frac{1}{2^n}\sum_xf(x)\chi_s(x)$ for $s\in\set{0,1}^n$,  where $\chi_s(x)\defeq(-1)^{s\cdot x}$. 
The function $f$ can be decomposed as $f=\sum_s\hat{f}(s)\chi_s$. The $\ell_p$ norm of $\hat{f}$ for any $p\geq 1$ is defined as $\norm{\hat{f}}_p\defeq\br{\sum_s\abs{\hat{f}(s)}^p}^{1/p}$. The Fourier sparsity $\norm{\hat{f}}_0$ is the number of nonzero Fourier coefficients of $f$.

Let $V\subseteq\F_2^n$ be an affine subspace and $f:V\rightarrow\F_2$ be a boolean function. A complexity measure of $f$ $m(f)$ is {\em downward non-increasing} if $m(f')\leq m(f)$ for any subfunction $f'$ obtained by restricting $f$ to an affine subspace of $V$. For instance, $\deg_2(\cdot)$ is downward non-increasing.
\begin{fact}\cite{TsangWXZ:2013}\label{fact:pdtleqdegM}
If $\rk{\cdot}\leq \mathrm{m}(\cdot)$ for some downward non-increasing complexity measure $m$, then it holds that $\pdt{f}\leq m(f)\cdot\log\norm{\hat{f}}_0$.
\end{fact}
\begin{fact}\cite{TsangWXZ:2013}\label{fact:rkleqminpcert}
For all non-constant $f: \F^n_2\rightarrow \F_2$, it holds  that $\rk{f}\leq\minpcert{}{f}.$
\end{fact}

\textbf{Additive combinatorics.} Given two sets $A, B\subseteq\F_2^n$ and an element $x\in\F_2^n$, $A+B\defeq\set{a+b: a\in A, b\in B }$ and $x+A\defeq\set{x+a: a\in A}$. For any integer $t$, $tA\defeq A+\ldots+ A$ where the summation includes $A$ for $t$ times.
Studying the structure of $tA$ for small constant $t$ is one of the central topics in additive combinatorics. Readers may refer to the excellent textbook~\cite{TaoVu:2009}. The following is the famous quasi-polynomial Bogolyubov-Ruzsa lemma due to Sanders~\cite{Sanders:2012}. It asserts that $4A$ contains a large subspace if $A\subseteq\F_2^n$ is large. Readers may refer to the nice exposition~\cite{Lovett:2014b} by Lovett.
\begin{fact}\cite{Sanders:2012, Ben-sassonRW:2014} \label{fact:bogloyubovRuzsa}
 Let $A\subseteq\F^n_2$ be a subset of size $|A|=\alpha2^n$. Then there exists a subspace $V$ of $\F^n_2$ satisfying $V\subseteq 4A$ and 
\[\codim{V}=\O\br{\log^4\br{\alpha^{-1}}}.\]
\end{fact}

\section{Main result}

\begin{lemma}\label{lem:BR}
Let $1\leq c\leq n$, $A_1,A_2,A_3,A_4\subseteq\F_2^n$ be subsets of size at least $2^{n-c}$. Then there exists an affine subspace $V\subseteq A_1+A_2+A_3+A_4$ of $\F_2^n$ such that 
\[\codim{V}=\O\br{c^4}.\]
\end{lemma}

\begin{proof}
The lemma is trivial if $c\geq n^{1/4}$. We assume that $c<n^{1/4}$. As $|A_1+ A_2|\leq 2^n$, there exists an element $a\in\F_2^n$ such that $a=a_1+ a_2$ for at least $2^{n-2c}$ pairs $(a_1,a_2)\in A_1\times A_2$. Then we have $|A_1\cap \br{A_2+ a}|\geq 2^{n-2c}$. For the same reason, there exists an element $a'\in\F_2^n$ such that $|A_3\cap \br{A_4+ a'}|\geq 2^{n-2c}$.  Note that $|\br{A_1\cap \br{A_2+ a}}+ \br{A_3\cap \br{A_4+ a'}}|\leq 2^n$. Thus there exists an element $a''\in\F_2^n$ such that $a''=a_3+ a_4$ for at least $2^{n-4c}$ pairs $(a_3,a_4)\in\br{A_1\cap\br{ A_2+ a}} \times \br{A_3\cap \br{A_4+ a'}}$. Set $$A=A_1\cap\br{A_2+ a}\cap\br{\br{A_3\cap\br{A_4+ a'}}+ a''}=A_1\cap\br{A_2+ a}\cap\br{A_3+ a''}\cap\br{A_4+ a'+ a''}.$$ We have $\abs{A}\geq2^{n-4c}>0$ since $c<n^{1/4}$. Thus there exists a subspace $V\subseteq 4A$ of codimension $\codim{V}\leq\O\br{c^4}$ by Theorem~\ref{fact:bogloyubovRuzsa}.  Note that $4A\subseteq A_1+ A_2+ A_3+ A_4+ a+ a'$. The affine subspace $V+ a+ a'$ serves the purpose.
\end{proof}
We define a downward non-increasing measure which is an upper bound on $\rk{\cdot}$.
\begin{definition}\label{def:M}
Given a function $f:V\rightarrow\F_2$, where $V$ is an affine subspace of $\F_2^n$ and $t\defeq\dim\br{V}$,  let $L:\F_2^t\rightarrow\F_2^n$ be an affine map satisfying that $L\br{\F_2^t}=V$. Set $F:\br{\F_2^t}^4\rightarrow\F_2$ by $F\br{x_1,x_2,x_3,x_4}\defeq f\br{L(x_1+x_2+x_3+x_4)}.$ The complexity of $f$ is defined to be $M(f)\defeq\cc{4}{F}.$
\end{definition}
\noindent Note that the affine map is invertible.  The complexity $M(f)$ does not depend on the choice of the affine map.
\begin{lemma}\label{lem:Mdownward}
$M\br{\cdot}$ is downward non-increasing.
\end{lemma}
\begin{proof}
Let $f:\F_2^n\rightarrow\F_2$ be a boolean function and $V\subseteq\F_2^n$ be an affine subspace. It suffices to show that $M\br{f}\geq M\br{f\textbar_{V}}$. Let $F$ and $F'$ be the $4$-argument functions given by Definition~\ref{def:M} corresponding to $f$ and $f\textbar_V$, respectively. Assume that $L\br{z}\defeq Az+b$ is the corresponding affine map in Definition~\ref{def:M}.  Given input $\br{x_1,x_2,x_3,x_4}\in\br{\F_2^t}^4$, where $t=\dim V$, player $P_1$ computes $x_1'=A_1x+b$ and players $P_i$ computes $x_i'=Ax_i$ for $i=2,3,4$.  Note that $L\br{x_1+x_2+x_3+x_4}=Ax_1+Ax_2+Ax_3+Ax_4+b$. We have  $F'\br{x_1,x_2,x_3,x_4}=f\br{x_1'+x_2'+x_3'+x_4'}=F\br{x_1',x_2',x_3',x_4'}$. The players simulate the protocol that computes $F$ on input $(x_1',x_2',x_3',x_4')$ and get $F'\br{x_1,x_2,x_3,x_4}$.   Thus $M\br{f'}=\cc{4}{F'}\leq\cc{4}{F}=M\br{f}$.
\end{proof}
\begin{lemma}\label{lem:minpcertvsM}
For any $f: V \rightarrow\F_2$, where $V$ is an affine subspace of $\F_2^n$, it holds that  $\minpcert{}{f}=\O\br{M\br{f}^4}$.
\end{lemma}
\begin{proof}
We assume w.l.o.g. that $V=\F_2^n$ . Let $F\br{x_1,x_2,x_3,x_4}\defeq f\br{x_1+ x_2+x_3+ x_4}$.  Let $c\defeq\cc{4}{F}$.  The optimal protocol partitions the domain into at most $2^c$ monochromatic hyperrectangles. Thus there exists a monochromatic hyperrectangle $A_1\times A_2\times A_3\times A_4$ satisfying $\abs{A_1\times A_2\times A_3\times A_4}\geq 2^{4n-c}$. Hence $\abs{A_i}\geq 2^{n-c}$ for $1\leq i\leq 4$. Using Lemma~\ref{lem:BR}, there exists an affine subspace $V\subseteq A_1+ A_2+A_3+A_4$ satisfying $\codim{V}=\O\br{c^4}$.  It implies that $\minpcert{}{f}\leq\O\br{c^4}$. The result follows. 
\end{proof}
Combining Fact~\ref{fact:rkleqminpcert}, Lemma~\ref{lem:Mdownward} and Lemma~\ref{lem:minpcertvsM}, we have
\[\pdt{f}\leq\O\br{M(f)^4\cdot\log\norm{\hat{f}}_0}.\]
By Definition~\ref{def:M}, $M\br{f}\leq\cc{4}{F}$. Note that $\log\norm{\hat{f}}_0\leq\cc{4}{F}$. The main theorem follows.

\subsection*{Open problems}

Here we list two open problems towards proving the Log-Rank Conjecture for XOR functions.

\begin{enumerate}

\item The most interesting work along this line is to show that the PDT complexity of $f$ and the communication complexity of the corresponding $2$-argument XOR-function $F_2$ are polynomially equivalent. 

\item Can we extend Theorem~\ref{thm:maintheorem} to the randomized communication complexity?
\end{enumerate}

\subsection*{Acknowledgement}
I would like to thank Ronald de Wolf for helpful discussion and improving the presentation. I also thank Shengyu Zhang for his comments and Shachar Lovett for informing us about his unpublished proof with Hatami.

\bibliographystyle{alpha}
\bibliography{references}

\appendix

\end{document}